\documentclass[submission,copyright,creativecommons]{eptcs}
\providecommand{\event}{GandALF 2019} 
\usepackage{breakurl}             
\usepackage{underscore}           

\usepackage{amsmath,amsfonts,amssymb,amsthm}
\usepackage{graphicx,color}
\usepackage[ruled,linesnumbered,algo2e,vlined]{algorithm2e}
\usepackage{setspace}
\usepackage{thm-restate}

\theoremstyle{theorem}
\newtheorem{theorem}{Theorem}

\newtheorem{corollary}{Corollary}
\theoremstyle{definition}

\newcommand{\mD}{\Sigma}
\newcommand{\emptystr}{\epsilon}
\newcommand{\None}{\textrm{null}}
\newcommand{\row}{\mathrm{row}}
\newcommand{\MQ}{\textrm{MQ}}
\newcommand{\Size}[1]{|#1|}
\newcommand{\Res}{\mathit{Res}}
\newcommand{\Prm}{\mathit{Prm}}
\newcommand{\dnt}[1]{[\![#1]\!]}
\newcommand{\DD}{\Phi}

\newcommand{\mrm}[1]{\mathrm{#1}}
\newcommand{\mcal}[1]{\mathcal{#1}}

\newcommand{\tsc}[1]{\textsc{#1}}

\newcommand{\Lstar}{\ensuremath{\textsf{L}*}}
\newcommand{\NLstar}{\ensuremath{\textsf{NL}*}}
\newcommand{\MATstar}{\ensuremath{\textsf{MAT}*}}

\SetCommentSty{mycommfont}

\title{Query Learning Algorithm for \\ Residual Symbolic Finite Automata}
\def\titlerunning{Query Learning Algorithm for Residual Symbolic Finite Automata}

\author{
	Kaizaburo Chubachi \email{kaizaburo\_chubachi@shino.ecei.tohoku.ac.jp}
\and
	Diptarama Hendrian \email{diptarama@tohoku.ac.jp}
\and
	Ryo Yoshinaka \email{ryoshinaka@tohoku.ac.jp}
\and
	Ayumi Shinohara \email{ayumis@tohoku.ac.jp}
\and\\[-30pt]
\institute{Graduate School of Information Sciences, Tohoku University, Japan}
}

\def\authorrunning{K. Chubachi, D. Hendrian, R. Yoshinaka and A. Shinohara}
\def\copyrightholders{Chubachi, Hendrian, Yoshinaka and Shinohara}

\begin{document}

\maketitle

\begin{abstract}
We propose a query learning algorithm for residual symbolic finite automata (RSFAs).
Symbolic finite automata (SFAs) are finite automata whose transitions are labeled by predicates over a Boolean algebra, in which a big collection of characters leading the same transition may be represented by a single predicate.
Residual finite automata (RFAs) are a special type of non-deterministic finite automata which can be exponentially smaller than the minimum deterministic finite automata and have a favorable property for learning algorithms.
RSFAs have both properties of SFAs and RFAs and can have more succinct representation of transitions and fewer states than RFAs and deterministic SFAs accepting the same language.
The implementation of our algorithm efficiently learns RSFAs over a huge alphabet and outperforms an existing learning algorithm for deterministic SFAs.
The result also shows that the benefit of non-determinism in efficiency is even larger in learning SFAs than non-symbolic automata.
\end{abstract}

\section{Introduction}

Learning regular languages has been extensively studied because of its wide varieties of applications in many fields such as pattern recognition, model checking, data mining and computational linguistics~\cite{Higuera05}.
Angluin~\cite{Angluin1987} presented an algorithm $\Lstar$ which learns the minimum deterministic finite automaton (DFA) accepting an unknown target language using \emph{membership queries (MQs)} and \emph{equivalence queries (EQs)}.
An MQ asks whether a string selected by the learner is a member of the target language or not.
An EQ asks whether the learner's hypothesis automaton accepts exactly the target language or not.
If not, the learner gets a counterexample from the symmetric difference of the hypothesis and target languages.
A teacher who can answer those two types of queries is called a \emph{minimally adequate teacher (MAT)}.
A number of different learning algorithms working under the MAT model have been designed for regular languages~\cite{IsbernerHS14,KearnsV94,RivestSchapire1993}.
These works have been found in applications such as specification generation~\cite{HowarS16, MargariaNRS04} and model verification~\cite{Leucker06}.
Recently, the algorithm is also used to extract a DFA representing behavior of recurrent neural networks~\cite{WeissGY18}.

In such applications, alphabets tend to be extremely large and structured.
The size of DFA representation grows linearly in the size of the alphabet
and the number of queries needed to learn a language over the alphabet also grows linearly.
Such difficulty can be alleviated by using \emph{symbolic finite automata (SFAs)}~\cite{VeanesHT10}.
An SFA has transitions that carry predicates over a Boolean algebra.
By using an algebra and its predicates suitable for the languages to represent,
we can make the representing SFA and learning processes for them more efficient.
For example, the edge from $q_0$ to $q_1$ is labeled with $6,7,8,9$ in the DFA in Fig.~\ref{fig:automata}(a), 
while the symbolic representation of it will be $\neg (X \leq 5)$ in Fig.~\ref{fig:automata}(b), 
 where $X$ is a free variable for which an input character is substituted.
One of the first query learning algorithms targeting some types of SFAs has been proposed by Mens and Maler~\cite{MensM15}.
Their algorithm assumes a stronger teacher than MAT, but it works efficiently over large ordered alphabets such as $\mathbb{N}$ or $\mathbb{R}$.
After that, several query learning algorithms which work under the standard MAT model have been proposed~\cite{Argyros2018,ArgyrosSKK16, DrewsD17, MalerM17}.
In particular, the algorithm $\MATstar$ given by Argyros and D'Antoni~\cite{Argyros2018} is quite generic.
It learns SFAs over any algebra when an efficient learning algorithm for the underlying algebra is available.
For example, as there exists a learning algorithm for binary decision diagrams (BDDs)~\cite{Nakamura05},
SFAs whose transitions carry BDDs can be learned by the algorithm.

Another way to represent a regular language compactly is to introduce non-determinism.
Denis et al.~\cite{Denis2002} have proposed \emph{residual finite automata (RFAs)}, which are a special kind of non-deterministic finite automata,
and presented nice properties of them including the fact that an RFA can be exponentially smaller than the minimum DFA accepting the same language.
Figure~\ref{fig:automata}(c) shows an RFA that accepts the same language as the DFA in Figure~\ref{fig:automata}(a).
Bollig et al.~\cite{Bollig2009} proposed an Angluin style learning algorithm $\NLstar$ for RFAs based on those nice properties
 and their experimental results demonstrated that $\NLstar$ needs fewer queries than $\Lstar$ in practice.

\begin{figure}[t]
	\centering
	\begin{minipage}[t]{0.4\hsize}
		\centering
		\includegraphics[scale=0.55]{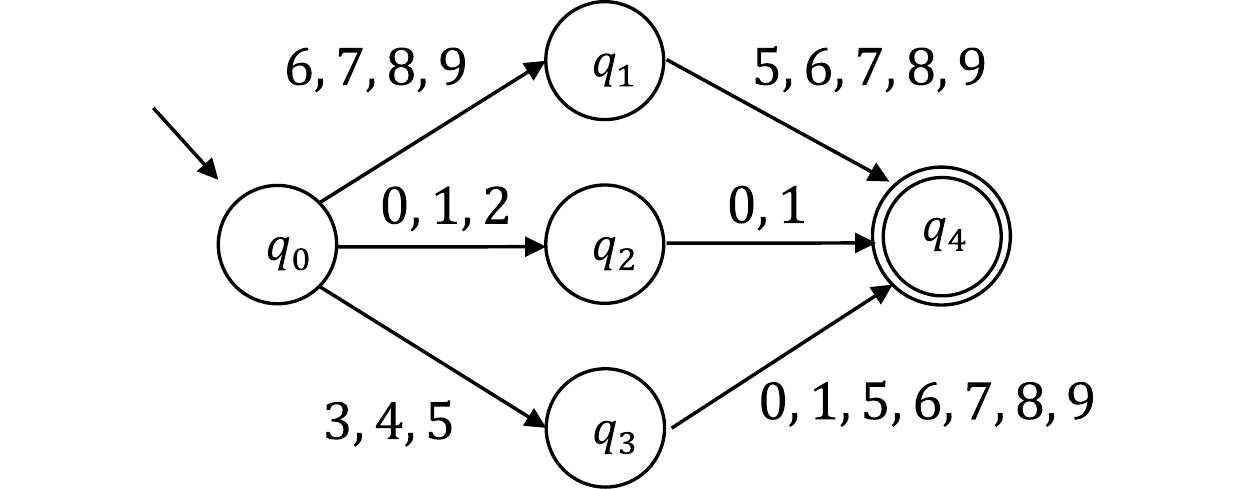}\\
		\ \ \ \small{(a) DFA \\ \ }
	\end{minipage}
	\quad\quad\quad
	\begin{minipage}[t]{0.5\hsize}
		\centering
		\includegraphics[scale=0.55]{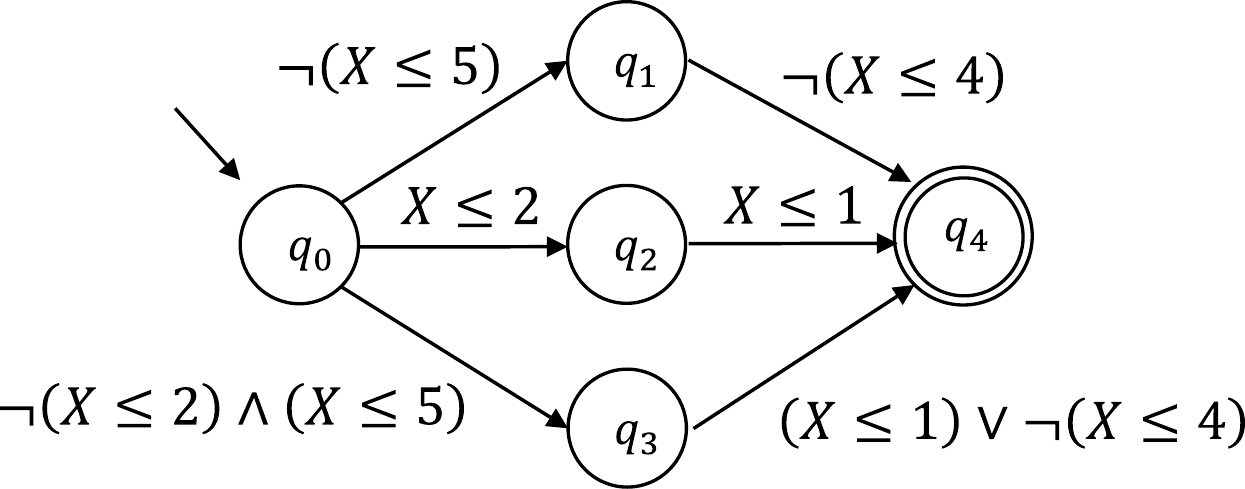}\\
		\ \ \ \small{(b) DSFA}
	\end{minipage}
	\\
	\begin{minipage}[t]{0.4\hsize}
		\centering
		\includegraphics[scale=0.55]{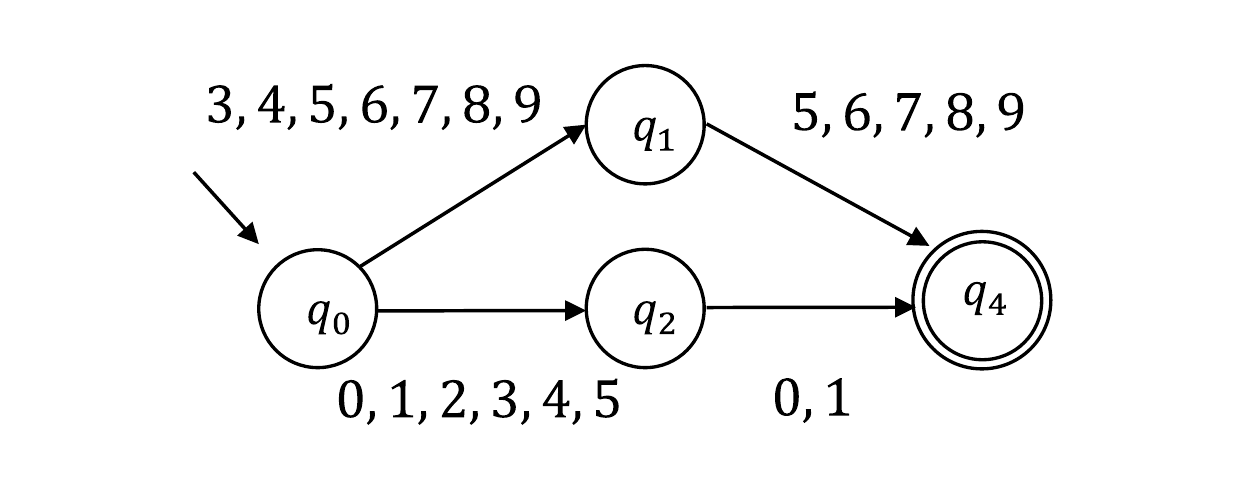}\\
		\ \ \ \small{(c) RFA \\ \ }
	\end{minipage}
	\quad\quad\quad
	\begin{minipage}[t]{0.5\hsize}
		\centering
		\includegraphics[scale=0.55]{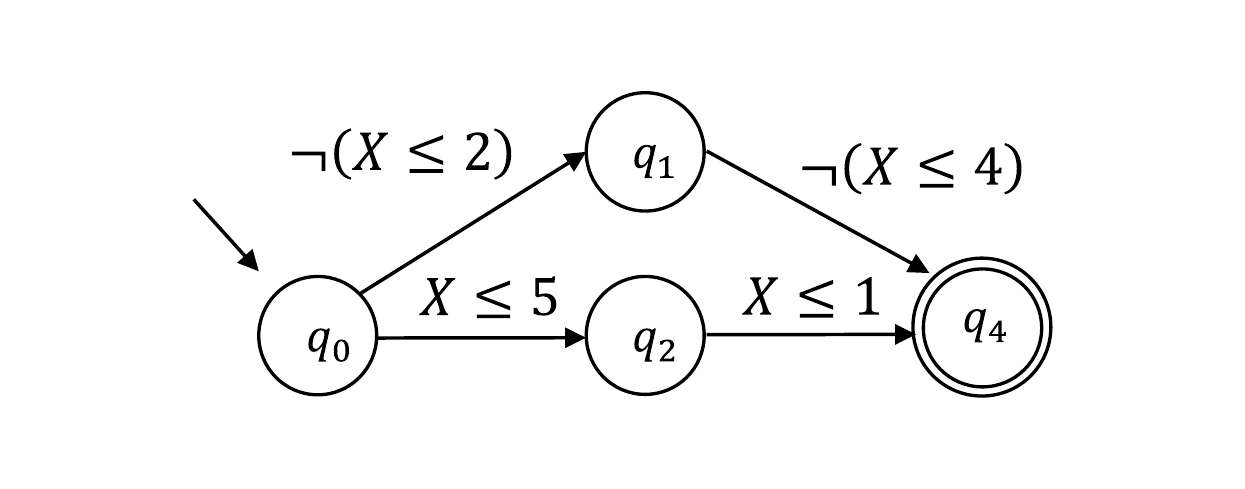}\\
		\ \ \ \small{(d) RSFA}
	\end{minipage}
	\caption{Examples of DFA, DSFA, RFA and RSFA over $\{a\in \mathbb{N}\mid 0\leq a \leq 9\}$. Every automaton accepts the same language and has the least number of states in each class.
	The dead state of the D(S)FA, from which one cannot reach the accepting state, is omitted.
	\label{fig:automata}}
\end{figure}

In this paper, we propose a learning algorithm for non-deterministic SFAs, which we call \emph{residual symbolic finite automata (RSFAs)} to pursue further compact representations and efficient learning.
Figure~\ref{fig:automata}(d) shows an example of an RSFA.
The algorithm can be seen as a combination of $\MATstar$~\cite{Argyros2018} and $\NLstar$~\cite{Bollig2009}.
We prove that our algorithm learns target languages using RSFAs under the MAT model
 and present upper bounds on the numbers of EQs and MQs required.
We also present experimental results that compare our algorithm with $\MATstar$.
We observe that the proposed algorithm asks much fewer EQs and MQs than $\MATstar$ like Bollig et al.~\cite{Bollig2009} have demonstrated for non-symbolic automata.
Yet, as we will discuss in more detail in Section~\ref{sec:experiments}, the impact looks significantly bigger in the learning of SFAs than non-symbolic automata.

As a byproduct of our algorithm analysis, we propose an improvement for $\NLstar$ that reduces the worst case query complexity.

%

\section{Preliminaries}

\subsection{Learning under Minimally Adequate Teacher}

Query learning (also called active learning) is a learning model where the learner constructs a representation of an unknown target language by actively asking queries about the language.
A most representative setting is learning under a \emph{minimally adequate teacher (MAT)}, proposed by Angluin~\cite{Angluin1987}.
For a target language $L_*\subseteq\Sigma^*$ over an alphabet $\Sigma$, a MAT answers two types of queries.
The first type is a \emph{membership query (MQ)}, whose instance is a string $w\in\Sigma^*$ selected by the learner.
The answer to an MQ on $w$, denoted by $\MQ(w)$, is $\MQ(w)=+$ if $w\in L_*$ and $\MQ(w)=-$ otherwise.
The second type is an \emph{equivalence query (EQ)}, whose instance is a hypothesis $\mathcal{H}$ that represents a language $L(\mcal{H})$.
The answer to an EQ is ``yes'' if the language $L(\mathcal{H})$ is equal to the target language $L_*$.
Otherwise, the answer is a counterexample from the symmetric difference $(L_* \setminus L(\mathcal{H})) \cup (L(\mathcal{H}) \setminus L_*)$ arbitrarily chosen by the teacher.
For convenience, we say that an algorithm \emph{learns a class $\mcal{R}$} of representations if it finally acquires a representation in $\mcal{R}$ that represents an arbitrary target language in $\{ L(\mcal{H}) \mid \mcal{H} \in \mcal{R}\}$.

Angluin has proposed a polynomial time algorithm $\Lstar$ for learning deterministic finite automata (DFA). 
Her algorithm $\Lstar$ has been improved in theoretical and practical query efficiency and memory efficiency~\cite{IsbernerHS14, RivestSchapire1993}.

\subsection{Residual Finite Automata}

An ($\epsilon$-free) non-deterministic finite automaton (NFA) is a quintuple $M=(\Sigma, Q, Q_0, F, \delta)$, where $\Sigma$ is a finite alphabet, $Q$ is a finite set of states, $Q_0\subseteq Q$ is the set of initial states, $F\subseteq Q$ is the set of final states, and $\delta \colon Q\times\Sigma \rightarrow 2^Q$ is the transition function.
An NFA is called deterministic (DFA) if $\Size{\delta(q, a)}=1$ for all $q\in Q$ and $a\in \Sigma$.
The transition function $\delta$ is extended to $\hat{\delta} \colon Q\times\Sigma^*\rightarrow 2^Q$ so that $\hat{\delta}(q, \emptystr)=\{q\}$ and $\hat{\delta}(q, aw)=\bigcup_{q'\in \delta(q, a)} \hat{\delta}(q', w)$ for $q\in Q$, $a\in\Sigma$ and $w\in\Sigma^*$.
We use $\delta$ to denote $\hat{\delta}$.
A string $w\in\Sigma^*$ is accepted by $M$ if $\delta(Q_0, w)\cap F\neq \emptyset$.
For each state $q\in Q$, the language \emph{accepted by $q$} is $L_q = \{w\in\Sigma^* \mid \delta(q, w) \cap F \neq \emptyset\}$.
The language accepted by $M$ is $L(M)=\bigcup_{q\in Q_0} L_{q}$.

A language $L'$ is a \emph{residual language} of $L$ if there exists $u\in\Sigma^*$ such that $L'= u^{-1}L= \{v\in\Sigma^* \mid uv\in L\}$.
The set of residual languages of $L$ is denoted by $\Res(L)$.
A \emph{residual finite automaton (RFA)}~\cite{Denis2002} is an NFA such that $L_q \in \Res(L(M))$ for every state $q\in Q$.
In other words, each state of an RFA accepts a residual language of $L(M)$.
It is not necessary that every residual language of $L(M)$ must be accepted by a single state.
A residual language $u^{-1}L(M)$ is the union of languages accepted by the states reached by reading $u$ from the initial states. 
Denis et al.~\cite{Denis2002} showed that an RFA can be exponentially smaller than the minimum DFA accepting the same language.

A language $L$ over a (possibly infinite) alphabet is called \emph{prime} in a class $\mathbb{L}$ of languages if it is not equal to the union of the languages it strictly contains, i.e., $L \neq \bigcup\{\, L' \in \mathbb{L} \mid L' \subsetneq L\,\}$.
The set of primes in $\mathbb{L}$ is denoted by $\Prm(\mathbb{L})$.
Denis et al.~\cite{Denis2002} showed that an RFA $M$
has the minimum number of states among RFAs accepting the same language if and only if $\Size{Q}=\Size{\Prm(\Res(L(M)))}$ and $L_q\in \Prm(\Res(L(M)))$ for each $q \in Q$.
Such an RFA is called \emph{reduced}, since no states can be deleted without changing its language.
Each state of the reduced RFA corresponds to a unique prime residual language of $L(M)$.
For a regular language $L$, the \emph{canonical} RFA of $L$ is $(\Sigma, Q, Q_0, F, \delta)$ where $Q=\Prm(\Res(L))$, $Q_0=\{L' \in Q \mid L' \subseteq L\}$, $F=\{L'\in Q \mid \emptystr\in L'\}$ and $\delta(L_1, a)=\{L_2 \in Q \mid L_2 \subseteq a^{-1}L_1\}$.
The canonical RFA is reduced and has saturated transitions (i.e.\ no transition can be added without modifying the language accepted by the RFA).

Bollig et al.~\cite{Bollig2009} has proposed an algorithm $\NLstar$ for learning RFAs extending Angluin's algorithm $\Lstar$ for DFAs.
It constructs (an RFA isomorphic to) the canonical RFA of the target language.
The theoretical upper bound on the number of queries required by $\NLstar$ is higher than $\Lstar$ for the same regular language.
However, their experimental results show that $\NLstar$ practically makes fewer queries than $\Lstar$ does.

\subsection{Symbolic Finite Automata}

\emph{Symbolic finite automata (SFAs)} are finite automata which have more expressive transitions than NFAs.
In an SFA, transitions carry unary predicates over an effective Boolean algebra $\mathcal{A}$ on a (typically huge or infinite) alphabet $\mD$.
Transitions whose predicates are satisfied by the read character $a\in\mD$ are executed.

An effective Boolean algebra is a tuple $\mathcal{A}=(\mD, \varPsi, \dnt{\_}, \bot, \top, \vee, \wedge, \neg)$, where $\mD$ is an alphabet, $\varPsi$ is a set of unary predicates closed under the Boolean connectives, and $\dnt{\_} \colon \varPsi \rightarrow 2^{\mD}$ is a denotation function such that (i) $\dnt{\bot}=\emptyset$, (ii) $\dnt{\top} = \mD$, and (iii) for all $\varphi, \psi \in \varPsi$, $\dnt{\varphi \vee \psi} = \dnt{\varphi} \cup \dnt{\psi}, \dnt{\varphi \wedge \psi} = \dnt{\varphi} \cap \dnt{\psi}$, and $\dnt{\neg\varphi} = \mD \setminus \dnt{\varphi}$.
We assume it is decidable whether $\dnt{\varphi} = \emptyset$ for any $\varphi \in \varPsi$ and moreover there is an effective procedure to find an element of $\dnt{\varphi}$ unless $\dnt{\varphi} = \emptyset$.

An SFA is a quintuple $M=(\mathcal{A}, Q, Q_0, F, \Delta)$, where $\mathcal{A}$ is an effective Boolean algebra, $Q$ is a finite set of states, $Q_0\subseteq Q$ is the set of initial states, $F\subseteq Q$ is the set of final states, $\Delta \subseteq Q \times \varPsi \times Q$ is the finite transition relation.
When the transition edge from $q$ to $q'$ has a predicate label $\varphi$, i.e., $(q, \varphi, q')\in \Delta$, this means that the transition is executed when $\varphi$ is satisfied by the reading character.
That is, $\Delta$ induces the transition function $\delta\colon Q\times\mD\rightarrow 2^Q$ such that $\delta(q, a)=\{q' \in Q \mid (q, \varphi, q')\in \Delta,\, a\in\dnt{\varphi}\}$ for $q\in Q$ and $a\in\mD$.
We extend $\delta$ to $\delta\colon Q\times\mD^*\rightarrow 2^Q$ in the same way as for (non-symbolic) FAs.
Without loss of generality, we may assume that each pair of states $q,q' \in Q$ has just one predicate $\varphi$ such that $(q,\varphi,q') \in \Delta$, since $\bot \in \varPsi$ and $\varPsi$ is closed under union.
Let $L_q=\{w\in\mD^* \mid \delta(q, w) \cap F \neq \emptyset \}$ for $q \in Q$.
The language $L(M)$ accepted by $M$ is $\bigcup_{q\in Q_0}L_q$.
An SFA is called deterministic if $|Q_0|=1$ and $\Size{\delta(q, a)}=1$ for all $q\in Q$ and $a\in \mD$.
SFAs inherit many virtues of (non-symbolic) FAs.
For example, one can effectively obtain the minimum DSFAs from SFAs and decide equivalence of two SFAs~\cite{VeanesHT10}.

Argyros and D'Antoni~\cite{Argyros2018} have given a MAT learner $\MATstar$ for deterministic SFAs (DSFAs), assuming that a MAT learner $\Lambda$ for $\varPsi$ is available.
That is, $\Lambda$ can learn $ \dnt{\varphi} \subseteq \mD$ for an arbitrary predicate $\varphi \in \varPsi$ with a MAT. 
The algorithm $\MATstar$  uses instances $\Lambda^{(q,q')}$ of $\Lambda$ to identify the predicate label of the transition edge from $q$ to $q'$
and pretends to be a MAT for those predicate learner instances.
Through communication between those predicate learners and the real MAT for the DSFA, it constructs a hypothesis DSFA.
Accordingly, the query complexity of $\MATstar$ depends on the design of $\Lambda$.
In general, it requires very much less MQs than classical MAT learners using DFAs, when the alphabet is big but finite.
If the alphabet is infinite, classical MAT learners have no hope to learn the language.

The target of this paper is \emph{residual symbolic finite automata (RSFAs)}.
An RSFA $M$ is an SFA such that $L_q\in \Res(L(M))$ for every state $q \in Q$.
It is called \emph{reduced} if $|Q|=|\Prm(\Res(L(M)))|$ and $L_q \in \Prm(\Res(L(M)))$ for each $q \in Q$.

\section{Learning Algorithm for Residual Symbolic Finite Automata}
\newcommand{\algMain}{
\begin{algorithm2e}[t]\small
	\caption{RFSA Learning algorithm\label{alg:main}}
	\SetVlineSkip{0.4mm}
	\SetKwFor{Loop}{loop}{}{end}
	initialize $\mathcal{T}\leftarrow (U, V, T)$ with $U=V=\{\emptystr\}$\;
	$\mathcal{H}\leftarrow\None$\;
	\Loop{}{
		\While{\upshape $\mathcal{H}$ is $\None$}{
			$\mathcal{H}\leftarrow\textrm{build\_hypothesis}(\mathcal{T})$ \tcp*[l]{Algorithm~\ref{alg:build-hypothesis}}
			$\mathcal{H}\leftarrow\textrm{confirm\_conditions}(\mathcal{T}, \mathcal{H})$ \tcp*[l]{Algorithm~\ref{alg:confirm-conditions}}
	}
		ask an EQ on $\mathcal{H}$\;
		\uIf{\upshape the teacher replies with a counterexample $w$}{
			$\mathcal{H}\leftarrow\textrm{process\_conterexample}(\mathcal{T}, \mathcal{H}, w)$ \tcp*[l]{ Algorithm~\ref{alg:process-counterexample}}
			$\mathcal{H}\leftarrow\textrm{confirm\_conditions}(\mathcal{T}, \mathcal{H})$ \tcp*[l]{Algorithm~\ref{alg:confirm-conditions}}
		}\lElse{%
			\Return $\mathcal{H}$ and terminate}
	}
\end{algorithm2e}
}

\newcommand{\algBH}{
\begin{algorithm2e}[t]\small
	\caption{\textrm{build\_hypothesis}$(\mathcal{T})$\label{alg:build-hypothesis}}
	\SetVlineSkip{0.4mm}
	$Q \leftarrow \{u\in U \mid \row(u)\in \Prm(\{\row(u') \mid u' \in U \})\,\}$\;
	$Q_0 \leftarrow \{u \in Q \mid \row(u) \subseteq \row(\emptystr)\}$\;
	$F \leftarrow \{u \in Q \mid \emptystr \in \row(u) \}$\;
	$\mathcal{H}\leftarrow (\mathcal{A}, Q, Q_0, F, \emptyset)$\;
	\For{$(q, q')\in Q \times Q$}{
		initialize the algorithm $\Lambda^{(q, q')}$\;
		$\mathcal{H}\leftarrow \textrm{update\_transition}(q, q', \Lambda^{(q, q')}, \mathcal{T}, \mathcal{H})$ \tcp*[l]{Algorithm~\ref{alg:update-transition}}
		\lIf{\upshape $\mathcal{H}$ is $\None$}{%
			\Return $\None$
		}
	}
	\Return $\mathcal{H}$\;
\end{algorithm2e}
}

\newcommand{\algUT}{
\begin{algorithm2e}[t]\small
	\caption{\textrm{update\_transition}$(q, q', \Lambda^{(q, q')}, \mathcal{T}, \mathcal{H})$\label{alg:update-transition}}
	\SetVlineSkip{0.4mm}
	\Repeat{\upshape $\Lambda^{(q, q')}$ asks an \tsc{eq} on a hypothesis $\varphi$}{
		$\Lambda^{(q, q')}$ asks an \tsc{mq} on $a\in\mD$\;
		$\mit{temp\_row} \leftarrow \{v\in V\mid \MQ(qav)=+ \}$\;\label{all:mq2}
		\If{\upshape $temp\_row \in \Prm(\{\row(u) \mid u \in U \} \cup \{temp\_row\}) \setminus \Prm(\{\row(u) \mid u \in U \})$}{
			extend $\mathcal{T}$ to $(U \cup \{ua\}, V, T')$ by MQs\;
			\Return \None\;
		}
		\lIf{\upshape $\row(q')\subseteq temp\_row$}{%
			answer the \tsc{mq} by $+$%
		}\lElse{%
			answer the \tsc{mq} by $-$%
		}
	}
	$\Delta'\leftarrow \{(q_1, \psi, q_2)\in \Delta \mid q_1 \neq q \vee q_2 \neq q'\} \cup \{(q, \varphi, q')\}$\;
	\Return $(\mathcal{A}, Q, Q_0, F, \Delta')$\;
\end{algorithm2e}
}

\newcommand{\algCC}{
\begin{algorithm2e}[h!]\small\setstretch{0.98}
	\caption{\textrm{confirm\_conditions}$(\mathcal{T}, \mathcal{H})$\label{alg:confirm-conditions}}	
	\SetVlineSkip{0.3mm}
	\lIf{\upshape $\mathcal{H}$ is $\None$}{\Return \None}
	\tcp{Confirm Condition~1}
	\For{\upshape $(q, q')\in Q^2$ such that $\row(q) \subseteq \row(q')$}{
		\For{$x \in Q$}{
			find $\varphi, \varphi'$ such that $(q, \varphi, x), (q', \varphi', x)\in \Delta$\;
			\If{\upshape $[\![\varphi \wedge \neg\varphi']\!]\neq \emptyset$}{
				find $a\in[\![\varphi \wedge \neg\varphi']\!]$\;
				\uIf{\upshape $\row(x)\not\subseteq \{v\in V\mid \MQ(qav)=+ \}$}{
					give $a$ to $\Lambda^{(q, x)}$ as counterexample\;
					$\mathcal{H}\leftarrow\textrm{update\_transition}(q, x, \Lambda^{(q, x)}, \mathcal{T}, \mathcal{H})$ \tcp*[l]{Algorithm~\ref{alg:update-transition}}
					\Return $\textrm{confirm\_conditions}(\mathcal{T}, \mathcal{H})$ \tcp*[l]{Algorithm~\ref{alg:confirm-conditions}}
				}
				\uElseIf{\upshape $\row(x)\subseteq \{v\in V\mid \MQ(q'av)=+ \}$}{
					give $a$ to $\Lambda^{(q', x)}$ as counterexample\;
					$\mathcal{H}\leftarrow\textrm{update\_transition}(q', x, \Lambda^{(q', x)}, \mathcal{T}, \mathcal{H})$ \tcp*[l]{Algorithm~\ref{alg:update-transition}}
					\Return $\textrm{confirm\_conditions}(\mathcal{T}, \mathcal{H})$ \tcp*[l]{Algorithm~\ref{alg:confirm-conditions}}
				}
				\Else{
					find $v\in V$ s.t.\ $v\in\row(x)$ and $\MQ(q'av)=-$ and 
					extend $\mathcal{T}$ to $(U, V \cup \{av\}, T')$ by MQs\;
					\Return $\None$\;
				}
			}
		}
	}
	\tcp{Confirm Condition~2}
	\For{\upshape $u \in U\setminus\{\emptystr\}$ in length ascending order}{
		\For{\upshape $x \in \delta(Q_0, u)$}{
			\If{\upshape $\row(x)\not\subseteq\row(u)$}{
				find $x'\in \delta(Q_0, u')$ such that $x\in\delta(x', a)$, where $u=u'a$ with $a\in\mD$\;
				\uIf{\upshape $\row(x)\not\subseteq \{v\in V \mid \MQ(x'av)=+ \}$}{
					give $a$ to $\Lambda^{(x', x)}$ as counterexample\;
					$\mathcal{H}\leftarrow\textrm{update\_transition}(x', x, \Lambda^{(x', x)}, \mathcal{T}, \mathcal{H})$ \tcp*[l]{Algorithm~\ref{alg:update-transition}}
					\Return $\textrm{confirm\_conditions}(\mathcal{T}, \mathcal{H})$ \tcp*[l]{Algorithm~\ref{alg:confirm-conditions}}
				}\Else(\tcp*[h]{We have $\row(x')\subseteq\row(u')$ and $\row(x'a)\not\subseteq\row(u'a)$}){%
					find $v\in V$ s.t.\ $\MQ(x'av)=+$ and $v\notin\row(u)$ and 
					extend $\mathcal{T}$ to $(U, V \cup \{av\}, T')$ by MQs\;
					\Return $\None$\;
				}
			}
		}
	}
	\tcp{Confirm Condition~3}
	\For{\upshape $v \in V\setminus\{\emptystr\}$}{
		\If{\upshape $\exists q\in Q, v\in\row(q) \Leftrightarrow v \notin L_q$}{
			find $v'\in\mD^*$, $a\in\mD$ and $q_2 \in Q$ such that \\ \quad $av'$ is suffix of $v$, $\forall q_1\in Q, \MQ(q_1v')=+ \Leftrightarrow v'\in L_{q_1}$ and $\MQ(q_2av')=+ \Leftrightarrow av'\notin L_{q_2}$\;\label{all:mq5}
			$temp\_row \leftarrow \{v'' \in V \mid \MQ(q_2av'') = + \}$\;
			\For{\upshape $q_3\in Q$}{
				\If{$\row(q_3)\subseteq temp\_row \Leftrightarrow q_3 \notin \delta(q_2, a)$}{
					give $a$ to $\Lambda^{(q_2, q_3)}$ as a counterexample\;
					$\mathcal{H}\leftarrow\textrm{update\_transition}(q_2, q_3, \Lambda^{(q_2, q_3)}, \mathcal{T}, \mathcal{H})$ \tcp*[l]{Algorithm~\ref{alg:update-transition}}
					\Return $\textrm{confirm\_conditions}(\mathcal{T}, \mathcal{H})$ \tcp*[l]{Algorithm~\ref{alg:confirm-conditions}}
				}
			}
			\uIf{\upshape $\MQ(q_2av') = -$ or $(\exists u\in U, \row(u)=temp\_row \wedge \MQ(uv')=+)$}{
				extend $\mathcal{T}$ to $(U, V \cup \{v'\}, T')$ by MQs\;
			}\lElse{%
				extend $\mathcal{T}$ to $(U \cup \{q_2a\}, V \cup \{v'\}, T')$ by MQs}
			\Return $\None$\;
		}
	}
	
	\Return $\mathcal{H}$\;
\end{algorithm2e}
}

\newcommand{\algPC}{
\begin{algorithm2e}[t]\small
	\caption{\textrm{process\_counterexample}$(\mathcal{T}, \mathcal{H}, w)$\label{alg:process-counterexample}}
	\SetVlineSkip{0.4mm}
	\eIf{\upshape $\bigvee_{q \in Q_0} {\MQ(qw)} \neq \MQ(w)$}{\label{all:mq3}
		extend $\mathcal{T}$ to $(U, V\cup\{w\}, T')$ by MQs\;\label{alg:ce-extend-1}
		\Return $\None$\;
	}{
		\scalebox{0.98}[1]{find $u, v\in\mD^*$ and $a\in\mD$ s.t.\ $w=uav$ and $\bigvee_{q \in \delta(Q_0, u)} \MQ(qav) \neq \bigvee_{q' \in \delta(Q_0, ua)} {\MQ(q'v)}$}\;\label{all:mq4}
		\eIf{\upshape $\bigvee_{q \in \delta(Q_0, u)} {\MQ(qav)}=+$}{
			find $q\in\delta(Q_0, u)$ such that $\MQ(qav)=+$\;
			$temp\_row \leftarrow \{v'\in V\mid MQ(qav')=+ \}$\;
			\eIf{\upshape $\exists q'\in Q, \row(q')\subseteq temp\_row \wedge q'\notin\delta(q, a)$}{
				give $a$ to $\Lambda^{(q, q')}$ as a counterexample\;
				\Return $\textrm{update\_transition}(q, q', \Lambda^{(q, q')}, \mathcal{T}, \mathcal{H})$ \tcp*[l]{Algorithm~\ref{alg:update-transition}}
			}{
				\uIf{\upshape $\exists u'\in U, \row(u') = temp\_row \wedge \MQ(u'v)=+$}{%
					extend $\mathcal{T}$ to $(U, V\cup\{v\}, T')$\;\label{alg:ce-extend-2}
				}\lElse{%
					extend $\mathcal{T}$ to $(U\cup\{qa\}, V\cup\{v\}, T')$\label{alg:ce-extend-3}}
				\Return $\None$\;
			}
		}{
			find $q\in\delta(Q_0, u)$ and $q'\in\delta(q, a)$ such that $\MQ(qav)=-$ and $\MQ(q'v)=+$\;
			\eIf{\upshape $\row(q')\not\subseteq \{v'\in V\mid \MQ(qav')=+ \}$}{
				give $a$ to $\Lambda^{(q, q')}$ as a counterexample\;
				\Return $\textrm{update\_transition}(q, q', \Lambda^{(q, q')}, \mathcal{T}, \mathcal{H})$ \tcp*[l]{Algorithm~\ref{alg:update-transition}}
			}{
				extend $\mathcal{T}$ to $(U, V\cup\{v\}, T')$\;\label{alg:ce-extend-4}
				\Return $\None$\;
			}
		}
	}
\end{algorithm2e}
}

\algMain

Our learning algorithm for RSFAs can be seen as a combination of 
the RFA learner $\NLstar$~\cite{Bollig2009}
and the DSFA leaner $\MATstar$~\cite{Argyros2018}.
This section presents how those can be combined and
 how the new difficulties raised by the combined setting shall be solved.

Our algorithm uses an \emph{observation table}~\cite{Angluin1987}, which is used by $\NLstar$.
An observation table is $\mathcal{T}=(U, V, T)$ where $U$ is a prefix-closed set of strings, $V$ is a set of strings, and $T$ is a map $T \colon U V\rightarrow \{+, -\}$.
For each $u\in U$ and $v\in V$, we make $T(uv)=+$ if $uv\in L_*$ and $T(uv)=-$ otherwise for all $u\in U$ and $v\in V$ by asking an MQ on the string $uv$.
An observation table can be viewed as a two-dimensional table whose $(u, v)$ entry is $T(uv)$ for $u\in U$ and $v \in V$.
Let $\row(u) = \{\, v \in V \mid T(uv)=+ \,\}$ for $u\in U$.
To observe the relationship among residual languages is a key to acquire a (reduced) RFA for the target language $L_*$, but we cannot directly handle $u^{-1} L_*$.
However, we can have a finite approximation $\row(u) = (u^{-1} L_*) \cap V$.
The algorithm builds a hypothesis based on this information.
Compared to the one constructed by $\NLstar$,
our observation table differs in two points:
(1) the domain of $T$ is $U V$ rather than $(U\cup U\mD) V$ and
(2) $V$ is not necessarily suffix-closed.
The first change is inevitable to handle huge alphabets.
The second is an improvement from $\NLstar$.
Giving up the idea of keeping $V$ suffix-closed reduces the size of $V$ and consequently the number of MQs.

Our algorithm builds a hypothesis SFA using the observation table and instances of the predicate learning algorithm, checks necessary conditions of the hypothesis to be an RSFA, asks an EQ and then updates the hypothesis by modifying the observation table and/or talking with the predicate learners.
This procedure is repeated until an EQ is answered ``yes''.
The pseudo-code of our algorithm is shown in Algorithm~\ref{alg:main}.
The hypothesis is rebuilt from scratch when the observation table is updated.
Our algorithm assigns $\None$ to the variable $\mathcal{H}$ if the hypothesis has to be rebuilt.

At the beginning of the main loop, our algorithm calls Algorithm~\ref{alg:build-hypothesis} to build a hypothesis $\mathcal{H}=(\mathcal{A}, Q, Q_0, F, \Delta)$ with $Q = \{u\in U \mid \row(u)\in \Prm(\{\row(u') \mid u' \in U \})\,\}$, $Q_0 = \{u \in Q \mid \row(u) \subseteq \row(\emptystr)\}$ and $F = \{u \in Q \mid \emptystr \in \row(u) \}$, using the observation table $\mathcal{T}=(U, V, T)$.
In order to construct $\Delta$, similarly to $\MATstar$, we use the MAT learning algorithm $\Lambda$ for the underlying algebra $\varPsi$.
After Algorithm~\ref{alg:build-hypothesis} creates $|Q|^2$ instances $\Lambda^{(q, q')}$ of $\Lambda$ for all pairs of states $q,q' \in Q$,
Algorithm~\ref{alg:update-transition} communicates with each $\Lambda^{(q, q')}$ by pretending to be a MAT to obtain a transition predicate from $q$ to $q'$.
To avoid confusion with EQs and MQs {from our algorithm to the MAT}, we use small capital letters \tsc{eq} and \tsc{mq} for equivalence queries and membership queries {from a predicate learner to our algorithm}, respectively.
When $\Lambda^{(q, q')}$ asks an \tsc{mq} on $a\in\mD$, our algorithm answers $+$ if $\row(q') \subseteq \{ v \in V  \mid \mathrm{MQ}(qav) = + \}$ and $-$ otherwise.
When $\Lambda^{(q, q')}$ asks an \tsc{eq} on $\varphi \in \varPsi$, our algorithm determines the transition predicate from $q$ to $q'$ to be $\varphi$. 
An answer to the \tsc{eq} from the predicate learner will be generated by analyzing the built transitions or a counterexample for an EQ on the hypothesis automaton in the following steps.
Execution of the predicate learner is suspended until a counterexample for the \tsc{eq} is found.
When the algorithm is trying to answer an \tsc{mq} on $a$ from a predicate learner $\Lambda^{(q,q')}$, if $\{ v \in V  \mid \mathrm{MQ}(qav) = + \}$ happens to be a new prime in $\{ u^{-1}L_* \cap V \mid u \in U \cup \{qa\} \}$, $\mathcal{T}$ is extended to $(U \cup \{qa\}, V, T')$, and the procedure of building hypothesis restarts from the beginning.

\algBH

\algUT

\algCC

\algPC

After building a hypothesis, we check the following three conditions on the hypothesis by Algorithm~\ref{alg:confirm-conditions}
and modify the hypothesis if necessary, before raising an EQ.
Those conditions ensure that our final output hypothesis will be a reduced RSFA.
\begin{itemize}
	\item\emph{Condition 1}: For $q, q'\in Q$ and $a\in\mD$, $\row(q) \subseteq \row(q')$ implies $\delta(q, a)\subseteq\delta(q', a)$.
	\item\emph{Condition 2}: For $u\in U$ and $x\in\delta(Q_0, u)$, we have $\row(x)\subseteq\row(u)$.
	\item\emph{Condition 3}: For $q\in Q$ and $v\in V$, we have $v \notin\row(q)$ iff $v\notin L_q$.
\end{itemize}
Note that, in $\NLstar$, the automaton derived from an observation table always satisfies essentially the same conditions as above, which ensures that $\NLstar$ finally acquires the canonical RSA for the learning target language.
The difference comes from the fact that $\NLstar$ makes a transition so that $q' \in \delta(q,a)$ if and only if $\row(q') \subseteq \{ v \in V  \mid \mathrm{MQ}(qav) = + \}$.
On the other hand, the predicate $\varphi$ on the transition edge from $q$ to $q'$ is determined by $\Lambda^{(q, q')}$ in our setting,
which ensures no special properties required to have the conditions.
Therefore, we need additional processes to check the three conditions. 

If some of the conditions is not satisfied, we modify the observation table or the transition relation so that the conditions shall be satisfied.
The transition relation is updated by giving a counterexample to a predicate learner's \tsc{eq} and receiving a new predicate hypothesis from it.
For Condition~1, recall that $\delta(q, a)\subseteq\delta(q', a)$ for all $a \in \mD$ if and only if $(q,\varphi,x),(q',\varphi',x)  \in \Delta$ implies $\dnt{\varphi} \subseteq \dnt{\varphi'}$, i.e., $\dnt{\varphi \wedge \neg \varphi'} = \emptyset$, for all $x \in Q$. 
By assumption, this can be confirmed effectively and when $\delta(q, a)\nsubseteq \delta(q', a)$, one can find a witness $a \in \dnt{\varphi \wedge \neg \varphi'}$.
Condition~2 can be checked by naively executing transitions reading all $u\in U$ from initial states.
By performing this in length ascending order, if Condition~2 fails,
one can find $u = u'a \in U$ with $a \in \mD$, $x' \in \delta(Q_0,u')$ and $x\in\delta(x', a)$ such that $\row(x')\subseteq \row(u')$ and $\row(x)\nsubseteq \row(u)$, thanks to the prefix-closedness of $U$.
Condition~3 can also be checked by naively executing transitions reading all $v\in V$ from each $q \in Q$.
Note that, since $V$ is not necessarily suffix-closed, differently from the previous case, we do not perform this in the length ascending order.
If some $v\in V$ is found to falsify the condition, i.e.\ $\exists q\in Q, v\in \row(q) \Leftrightarrow v\notin L_q$,
we find a suffix $av'$ of $v$ with $a\in\mD$, which is not necessarily in $V$, such that $\mrm{MQ}(qv')=+ \Leftrightarrow v' \in L_q$ for all $q\in Q$ and $\mrm{MQ}(q_2 av')=+ \Leftrightarrow av' \notin L_{q_2}$ for some $q_2 \in Q$.
Such a suffix $av'$ can be found with $O(\Size{Q}\log \Size{v})$ MQs using binary search on the suffixes of $v$,
since  $\mrm{MQ}(q \emptystr)=+ \Leftrightarrow \emptystr \in L_q$ for all $q\in Q$ and $\mrm{MQ}(q_2 v)=+ \Leftrightarrow v \notin L_{q_2}$ for some $q_2 \in Q$.
Then, using such $a \in \mD$ and $q_2 \in Q$, we modify the hypothesis.
Note that by employing this technique, the query efficiency of $\NLstar$~\cite{Bollig2009}, which requires $V$ to be suffix-closed, can be improved (Corollary~\ref{cor:RFA}).

When the hypothesis is confirmed to satisfy the three conditions above, the algorithm asks an EQ.
We prove in Section~\ref{section:proofs} that if the hypothesis passes the equivalence test, it is a reduced RSFA. 

When a counterexample $w$ is given to an EQ, we process the counterexample.
Algorithm~\ref{alg:process-counterexample} is a modification of $\MATstar$ for our non-deterministic hypothesis.
At First, we check whether $\bigvee_{q \in Q_0} {\MQ(qw)} = \MQ(w)$.
If not, $Q_0$ shall be refined by adding $w$ to $V$.
If it is the case, we find a decomposition of $w = uav$ such that $u, v\in \mD^*$, $a\in\mD$ and $\bigvee_{q \in \delta(Q_0, u)} {\MQ(qav)} \neq \bigvee_{q' \in \delta(Q_0, ua)} {\MQ(q'v)}$.
We have $\bigvee_{q \in Q_0} {\MQ(qw)} = \MQ(w)$ and $\bigvee_{q \in \delta(Q_0, w)} {\MQ(q)} \neq \MQ(w)$ because $\bigvee_{q \in \delta(Q_0, w)} {\MQ(q)} = + \Leftrightarrow \delta(Q_0, w)\cap F \neq \emptyset$ and $w$ is a counterexample such that $\delta(Q_0, w)\cap F \neq \emptyset \Leftrightarrow \MQ(w) = -$.
This implies $\bigvee_{q \in \delta(Q_0, \emptystr)} {\MQ(qw)} \neq \bigvee_{q \in \delta(Q_0, w)} {\MQ(q)}$.
Therefore, such a decomposition can be found with $O(\Size{Q} \log \Size{w})$ MQs using binary search on the decompositions of $w$.
This is a non-deterministic extension of the binary search technique in~\cite{RivestSchapire1993} for learning DFAs.
Then, we can refine the transition from $q$ led by $a$.

\section{Correctness and Termination}
\label{section:proofs}

Although our hypothesis $\mathcal{H}$ is not guaranteed to be always an RFSA,
the algorithm will eventually terminate and return a reduced RSFA accepting the target language.



\begin{restatable}{theorem}{thmB}
\label{thm:correctness}
When the hypothesis $\mcal{H}$ passes the equivalence test, $\mathcal{H}$ is a reduced RSFA accepting the target language $L_*$. 
\end{restatable}
One can prove Theorem~\ref{thm:correctness} in the essentially same manner as for $\NLstar$~\cite{Bollig2009}.

In order to evaluate the query complexity of our algorithm, we first discuss how many \tsc{mq}s and \tsc{eq}s a predicate learner $\Lambda^{q,q'}$ may make.
Let $D = \{\, a \in \mD \mid \row(q') \subseteq \row(qa)\}$ be the set of letters $a \in \mD$ that $\Lambda^{q,q'}$ is expected to learn, in accordance with how our learner answers \tsc{mq}s from $\Lambda^{q,q'}$.
To find a predicate for $D$, we refer to the minimum DSFA $M_*=(\mcal{A},Q_*,q_0,F_*,\Delta_*)$ that accepts the learning target $L_*$.
Note that $M_*$ is constructible from an arbitrary SFA for $L_*$~\cite{VeanesHT10}.
Let $r = \delta_*(q_0,q)$ and $Q_*^D = \{\,\delta_*(r,a) \mid a \in D\,\}$, where $\delta_*$ is the transition function induced by $\Delta_*$.\footnote{To be strict, the DSFA should be written as $(\mathcal{A}, Q_*, \{q_0\}, F_*, \Delta_*)$ and $\delta_*(q_0,q)$ is a singleton set according to what we have defined in the preliminary section, but here we follow the conventional notation for deterministic automata.}  
Then we have $D = \dnt{\bigvee_{r' \in Q_*^D} \varphi_{r,r'}}$ where $\varphi_{r,r'}$ is the predicate on the edge from $r$ to $r'$. 
Therefore, $\Lambda$ is indeed capable of learning $D$ and there are bounds on the numbers of \tsc{eq}s and \tsc{mq}s that $\Lambda$ makes,
which we write as $\mathcal{C}_{\textrm{EQ}}({D})$ and $\mathcal{C}_{\textrm{MQ}}({D})$, respectively.

For each $L\in\Res(L_*)$, let $\Gamma_{L} = \{ \{ a \in \mD \mid L' = a^{-1}L \} \mid L' \in \Res(L_*)\}$.
Then, the set of denotations of predicates that may appear in an automaton built by the learner during the learning process is represented as
$\DD = \{ \bigcup_{D\in S} D \mid S \subseteq \Gamma_{L} \text{ for } L\in Res(L_*)\}$.
Then the numbers of \tsc{eq}s and \tsc{mq}s that each predicate learner may make are bounded by $\mathcal{E}=\max_{D \in \DD}\mathcal{C}_{\textrm{EQ}}(D)$ and $\mathcal{M}=\max_{D \in \DD}\mathcal{C}_{\textrm{MQ}}(D)$, respectively.
\begin{theorem}
\label{thm:complexity}
	Let $n = |\Res(L_*)|$ and $m$ be the length of the biggest counterexample to an EQ returned by the MAT.
	Then, the proposed algorithm returns a reduced RSFA accepting $L_*$ using $\Lambda$ after raising at most $O(n^4 \mathcal{E})$ EQs and $O(n^6(\mathcal{E} + \mathcal{M}) + n^5 \mathcal{E} \log m)$ MQs.
\end{theorem}

\begin{proof}
	At first, we will prove that the observation table $\mathcal{T}$ cannot be extended beyond $O(n^2)$ times.
	Following~\cite{Bollig2009}, we create a tuple $(l_U, l, p, i)$ of measures where $l_U=|\{\row(u) \mid u\in U\}|$, $l=|R|$, $R=\{\{v\in V\mid uav\in L_* \} \mid u\in U, a\in\mD\cup\{\emptystr\}\}$, $p=|\Prm(\{\row(u) \mid u\in U\})|$, $i=|\{(r, r')\mid r, r'\in R, r\subsetneq r'\}|$.
	After each extension of the table, either (1) $l_U$ is increased or (2) $l$ is increased by $k>0$ and, simultaneously, $i$ is increased by at most $kl+k(k - 1)/2$ or (3) $l$ stays the same and $i$ decreases or $p$ increase.
	However, $l_U$, $l$, $p$ cannot increase beyond $n$. Therefore, $\mathcal{T}$ cannot be extended beyond $O(n^2)$ times.

	Recall that $\Delta$ is updated only when a counterexample to some predicate learner is found.
	Such counterexamples are found at most $|\Delta|\mathcal{E}$ times without extending $\mathcal{T}$.
	$|\Delta|$ can be bounded by $O(n^2)$.
	Thus, $\Delta$ is updated at most $O(n^2\mathcal{E})$ times without extending $\mathcal{T}$.
	Therefore, The algorithm must always reach an EQ and terminate after making at most $O(n^4\mathcal{E})$ EQs.

	The algorithm asks MQs for 
	(1) filling the observation table after extending the table,
	(2) answering \tsc{mq}s from predicate learners at Line~\ref{all:mq2} of Algorithm~\ref{alg:update-transition},
	(3) checking $\bigvee_{q \in Q_0} {\MQ(qw)} \neq \MQ(w)$ for a counterexample $w$ at Line~\ref{all:mq3} of Algorithm~\ref{alg:process-counterexample},
	(4) finding a decomposition of a counterexample at Line~\ref{all:mq4} of Algorithm~\ref{alg:process-counterexample},
	(5) finding a suffix of $v\in V$ which does not satisfy Condition~3 at Line~\ref{all:mq5} of Algorithm~\ref{alg:confirm-conditions},
	and
	(6) the other purpose, where we check one or two rows for deciding whether to extend the observation table or to update the transition relation, after finding a decomposition of a counterexample or when one of the three conditions is found to be unsatisfied.

	The total number of MQs used for (1) is $O(n^3)$, because $|U|$ and $|V|$ is bounded by $n$ and $O(n^2)$, respectively.
	Concerning (2), we use at most $|\Delta||V|\mathcal{M}$ MQs for each intermediate observation table $\mathcal{T}$. Thus, $O(n^6\mathcal{M})$ MQs are asked in total.
	We perform (3) and (4) at most $O(n^4\mathcal{E})$ times. For each event, (3) uses $O(n)$ MQs and (4) uses $O(n\log m)$ MQs. All in all, $O(n^5\mathcal{E}\log m)$ MQs are asked for (3) and (4).
	Each time (5) happens, $O(n\log m)$ MQs are raised, and (5) happens $O(n^4\mathcal{E})$ times. In total, $O(n^5\mathcal{E}\log m)$ MQs are asked for (5).
	For each decision of (6), $O(|V|)$ MQs are used, and it takes place $O(n^4\mathcal{E})$ times. Hence, $O(n^6\mathcal{E})$ MQs are made for (6).
	Therefore, the algorithm asks at most $O(n^6(\mathcal{E} + \mathcal{M}) + n^5 \mathcal{E} \log m)$ MQs.
\end{proof}

A query complexity comparison among previous algorithms and our algorithm for related classes of automata is shown in Table~\ref{table:query-comparison}.
The query complexity of the proposed algorithm is higher than that of $\MATstar$~\cite{Argyros2018}, especially for MQs.
This is also true in the non-symbolic case.
However, Bollig et al.~\cite{Bollig2009} showed that in practice learning RFAs requires less queries than learning DFAs.
In Section~\ref{sec:experiments}, we show that it is also the case in the learning of RSFAs and DSFAs. 

We remark that the query efficiency of $\NLstar$~\cite{Bollig2009} can be improved by using our technique.
The algorithm $\NLstar$~\cite{Bollig2009} adds all the suffixes of a counterexample to $V$, which makes $V$ suffix-closed and rather large.
Suffix-closedness of $V$ ensures the correctness of $\NLstar$.
Namely, the property is used in the proof of Lemma~2 of~\cite{BolligHKL2008}, which states that Condition~3 of our paper always holds for $\NLstar$.
That is, by employing our counterexample processing and Condition~3 assurance procedure,
 the upper bound of the size of the table is improved from $O(n^3m\Size{\mD})$ to $O(n^3\Size{\mD})$ with additional $O(n^3 \log m)$ MQs (binary search with $O(n\log m)$ MQs can occur $O(n^2)$ times).
\begin{corollary}\label{cor:RFA}
	Canonical RFAs can be learned using $O(n^2)$ EQs and $O(n^3|\Sigma| + n \log m)$ MQs.
\end{corollary}

\begin{table}[t]
	\centering
	\caption{The upper bounds of EQs and MQs}
	\begin{tabular}{|cc|c|c|c|}
		\hline
		& & Deterministic & \multicolumn{2}{c|}{Residual} \\
		\hline\hline
		FA & EQ & $n$ & $O(n^2)$ & $O(n^2)$ \\
		& MQ & $O(n^2|\Sigma| + n\log m)$~\cite{IsbernerHS14,RivestSchapire1993} & $O(n^3m|\Sigma|)$~\cite{Bollig2009} & $O(n^3|\Sigma| + n^3 \log m)$~[Ours]\\
		\hline
		SFA & EQ & $O(n^3\mathcal{E})$ & \multicolumn{2}{c|}{$O(n^4\mathcal{E})$} \\
		& MQ & $O(n^4\mathcal{M} + n^4\mathcal{E}\log m)$~\cite{Argyros2018}& \multicolumn{2}{c|}{$O(n^6(\mathcal{E} + \mathcal{M}) + n^5\mathcal{E}\log m)$~[Ours]} \\
		\hline
	\end{tabular}
	\label{table:query-comparison}
\end{table}

\section{Experiments}\label{sec:experiments}

To evaluate the practical performance of our algorithm, we compare it
with Argyros and D'Antoni's\linebreak  algorithm $\MATstar$ for DSFAs~\cite{Argyros2018}.
They implemented their algorithm on the open-source library\linebreak \texttt{symbolicautomata}.
Our algorithm is also implemented on the same library.

\subsection{Setting}

We generated learning target languages as follows, which can be seen as the symbolic counterpart of the languages given by Denis et al.~\cite{Denis2004} for comparing DFAs and RFAs.
Denis et al.\ used NFAs over a two-letter alphabet for generating random regular languages.
We use non-deterministic (not necessarily residual) SFAs over the entire 32-bit integers, i.e. $\mD=\{a\in\mathbb{N} \mid -2^{31} \le a \le 2^{31} - 1 \}$.
SFAs we use are on the inequality algebra over $\mD$, whose atomic predicates are of the form $X \leq k$ for some $k \in \mD$, where $X$ is the free variable, and
their semantics is given by $\dnt{X \le k} = \{\, a \in \mD \mid a \le k \,\}$.
Using negation, intersection and union, predicates define unions of intervals.
We abbreviate $\neg(X\leq l)\wedge (X\leq r)$ to $l+1 \le X \le r$.
We also implemented a MAT learning algorithm for the algebra and used it as $\Lambda$. 
It learns arbitrary subsets $S \subseteq \mD$ by asking at most $K$ \tsc{eq}s and $O(K\log \Size{\mD})$ \tsc{mq}s using binary search
where $K$ is the number of ``borders'' $\Size{\{a\in\mD \mid a\in S \Leftrightarrow a - 1 \notin S \}}$ in the set $S$.
SFAs are randomly generated using four parameters:
the number $n_Q$ of states, the number $n_{\delta}$ of transitions per state, and the probabilities $p_I$ and $p_F$ for each state of being an initial and final state, respectively.
For each state $q\in Q$, we randomly pick a destination state $q'\in Q$ and two integers $l, r\in \mD$ such that $l \le r$ and add $(q,\ l \le X \le r,\ q')$ to $\Delta$.
This addition is performed $n_{\delta}$ times for each state permitting duplication of the destination state.
If we choose the same state $q'$ twice or more as a destination of a state $q$, the transition predicate from $q$ to $q'$ will be the union of two or more randomly chosen intervals.
We used the parameters $n_Q=8$, $n_{\delta}=2$ and $p_I=p_F=0.5$ in our experiments.\footnote{The source code is available at \url{https://github.com/ushitora/RSFA-QueryLearning}.}

\begin{figure}[t]
	\centering
	\begin{minipage}[t]{0.49\hsize}
		\centering
		\includegraphics[scale=0.45]{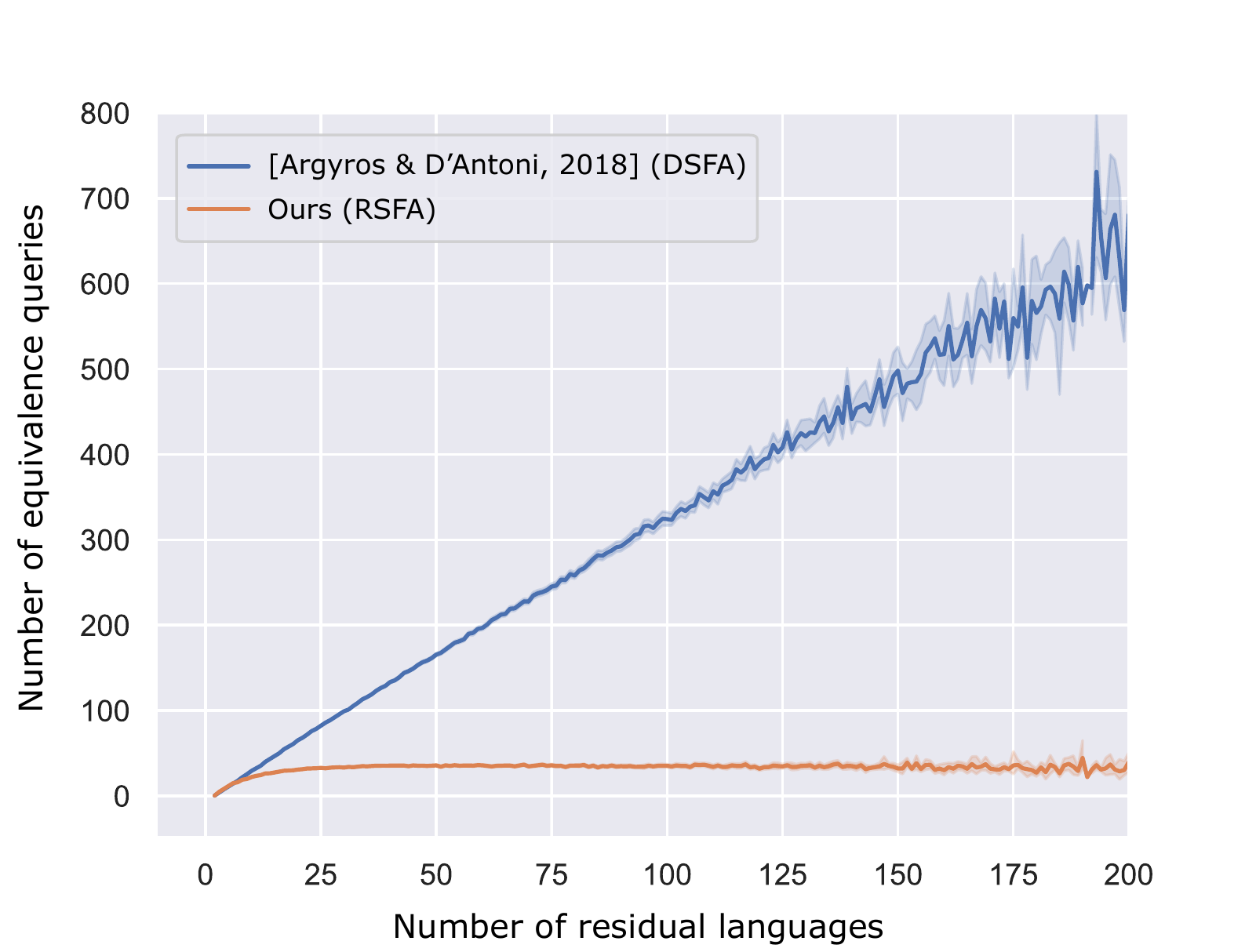}\\
		\ \ \ \scriptsize{(a)}
	\end{minipage}
	\begin{minipage}[t]{0.49\hsize}
		\centering
		\includegraphics[scale=0.45]{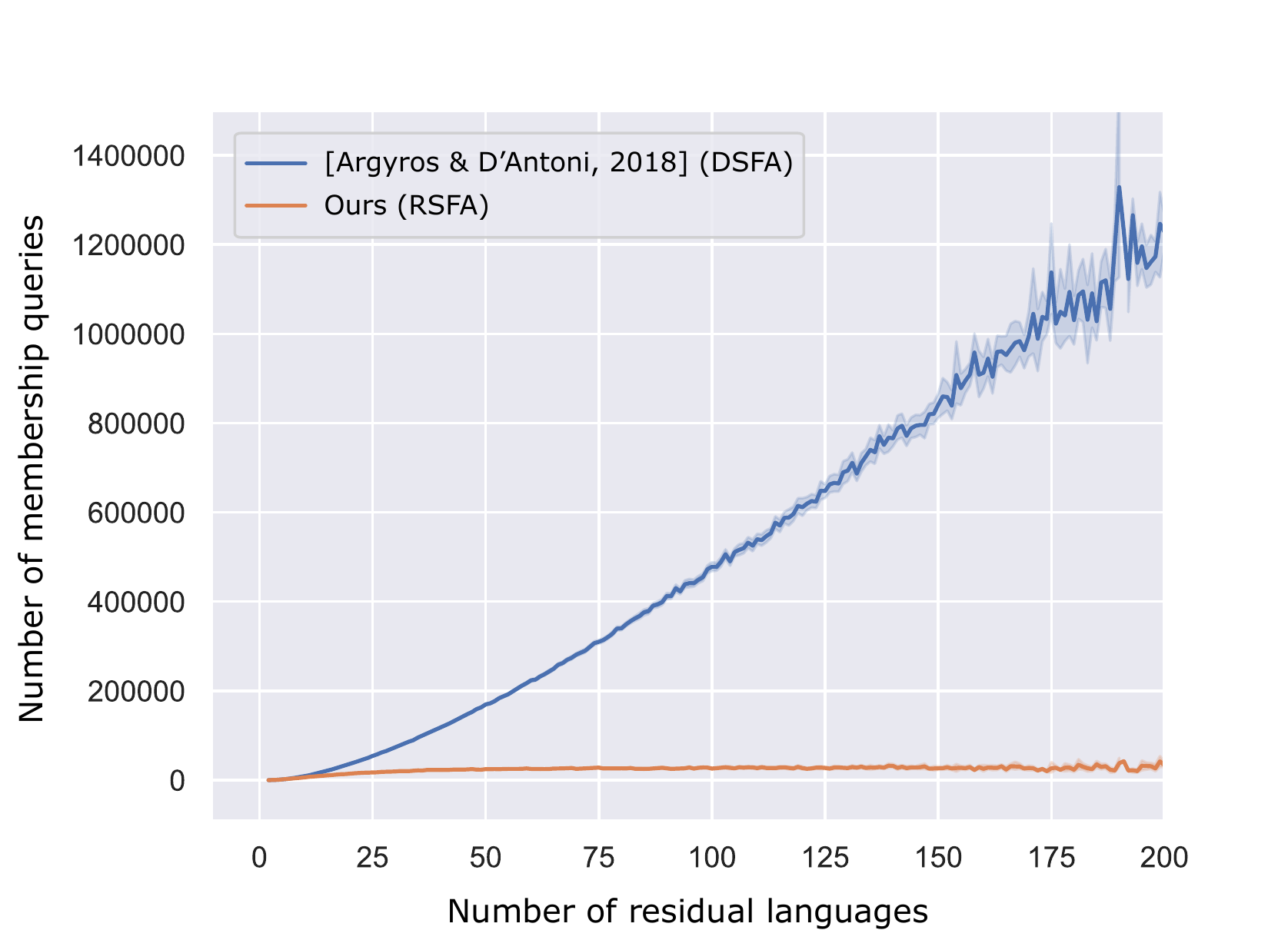}\\
		\ \ \ \scriptsize{(b)}
	\end{minipage}
	\caption{The average number of (a) EQs and (b) MQs relative to the number of the total residual languages, i.e. the size of the minimal DSFA. Error bands show a 95\% confidence interval.}
	\label{fig:experiment-int-interval}
\end{figure}

\subsection{Results}

We generated 50,000 non-deterministic SFAs and let our algorithm learn the languages defined by those SFAs.
Figures~\ref{fig:experiment-int-interval}(a) and (b) show the average numbers of EQs and MQs raised by our algorithm relative to the number of the residual languages, respectively.
The results are in contrast with our worst case analysis in Theorem~\ref{thm:complexity}.
Our algorithm makes much fewer queries than $\MATstar$.
The gap between the numbers of the queries made by $\MATstar$ and our algorithm looks even bigger than 
that between those by $\Lstar$ and $\NLstar$ observed in the experiments performed by Bollig et al.~\cite{Bollig2009} and by Angluin et al.~\cite{AngluinEF15}.
In the remainder of this section, we discuss why using non-deterministic version should be more beneficial in the learning of symbolic automata than non-symbolic automata.

\subsection{Analyses and Discussions}

Denis et al.~\cite{Denis2004} have observed that most languages of randomly generated NFAs have few prime residual languages, i.e.,
the number of states of a reduced RFA tends to be much fewer than that of the minimum DFA for the same language. 
Even in the middle of the learning process, hypotheses built by our learner tend to be smaller than the ones by $\MATstar$ (c.f.\ Fig.~\ref{fig:experiment-int-interval-V}(a)), in spite of the worst case analysis.
This tendency should be essentially the same in the non-symbolic and symbolic cases.
However, the automaton size has a bigger effect on the query complexity in the symbolic case than in the non-symbolic case.
Recall that most MQs and EQs to the MAT are used to answer \tsc{mq}s and \tsc{eq}s from the predicate learners when learning SFAs. 
We have $|Q|^2$ predicate learners if our current hypothesis has $|Q|$ states.
As a consequence, the benefit in the query complexity to reduce the number of states in a hypothesis automaton is much bigger in SFA learning
and thus using residual automata is quite advantageous.

In addition, when learning RSFAs, we are granted to be flexible to some extent in identification of transition predicates.
Concerning a transition from $q$ to $q'$, let $S_{\mathrm{saturated}}=\{a\in\mD \mid L_{q'} \subseteq a^{-1} L_{q} \}$ and $S_{\mathrm{simplified}}=\{a\in\mD \mid L_{q'} \subseteq a^{-1} L_{q}, (\nexists q'' \in Q, L_{q'}\subsetneq L_{q''} \subseteq a^{-1} L_q) \}$.
For any $\varphi$ with $S_{\mathrm{simplified}} \subseteq \dnt{\varphi} \subseteq S_{\mathrm{saturated}}$, changing the transition predicate between $q$ and $q'$ to $\varphi$ does not change the language defined by the automaton.  
That is, when learning an RSFA, as long as the predicate learner $\Lambda^{(q,q')}$ outputs a predicate $\varphi$ satisfying $S_{\mathrm{simplified}} \subseteq \dnt{\varphi} \subseteq S_{\mathrm{saturated}}$, the RFSA constructed using that output may pass the equivalence test.
For instance, in the inequality algebra, when $S_{\mathrm{saturated}}$ consists of many intervals, a ``lazy'' predicate whose semantics consists of fewer intervals may be accepted, which can be achieved with fewer queries.
This nature is not observed neither in RFAs nor DSFAs.

\begin{figure}[t]
	\centering
	\begin{minipage}[t]{0.49\hsize}
		\centering
		\includegraphics[scale=0.45]{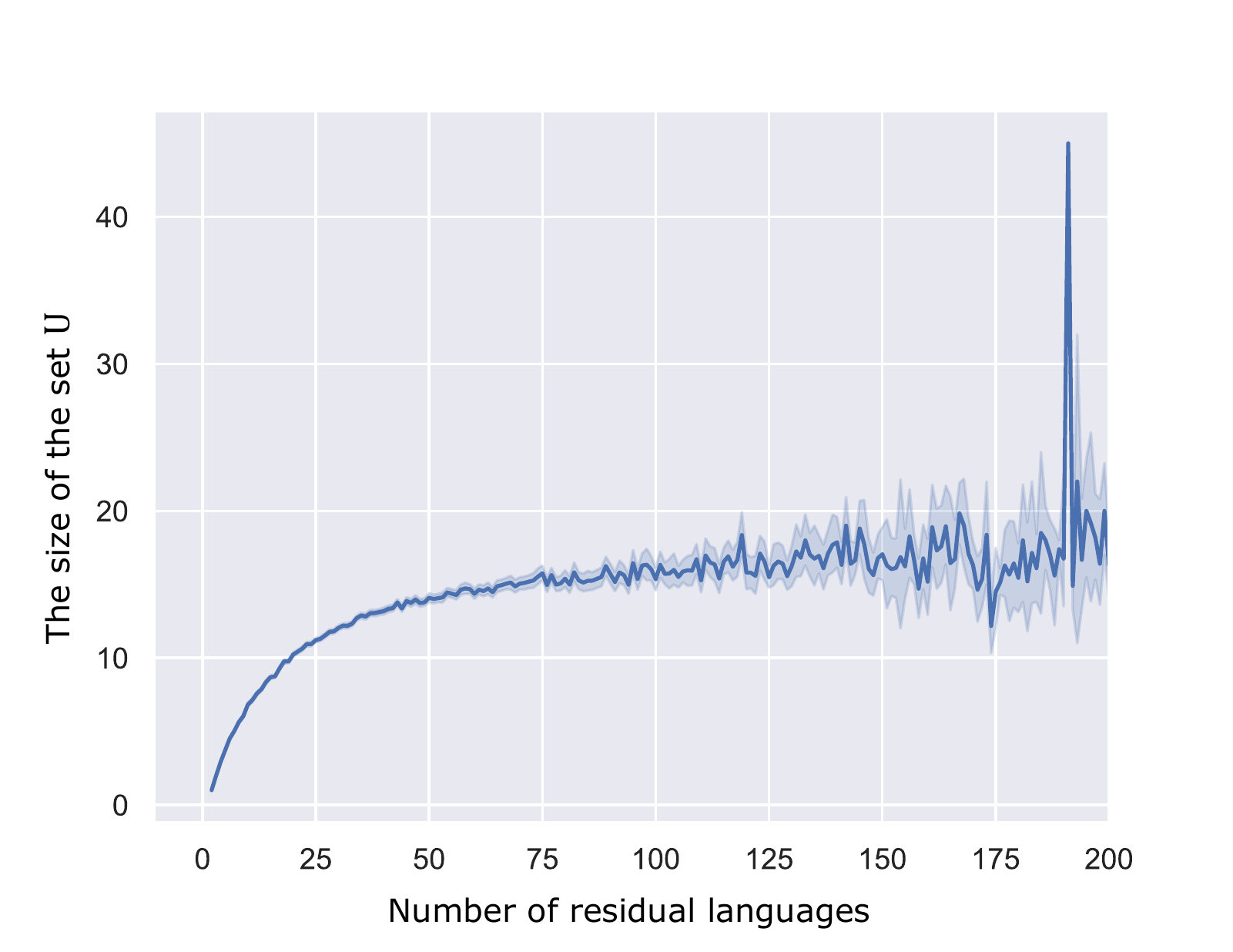}\\
		\ \ \ \scriptsize{(a)}
	\end{minipage}
	\begin{minipage}[t]{0.49\hsize}
		\centering
		\includegraphics[scale=0.45]{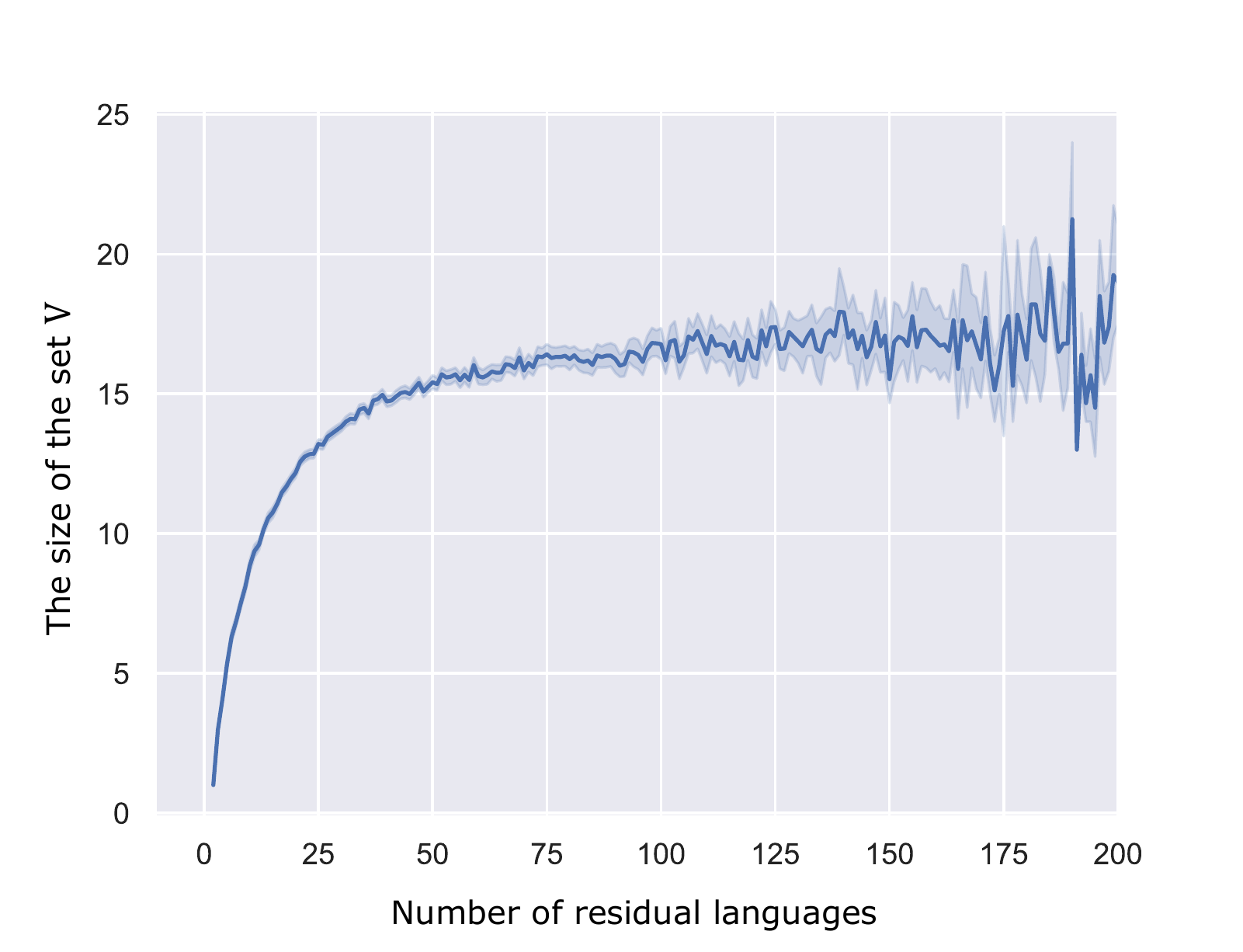}\\
		\ \ \ \scriptsize{(b)}
	\end{minipage}
	\caption{The average size of (a) $U$ and (b) $V$ relative to the number of the residual languages with error bands with a 95\% confidence interval. The number of states in an (intermediate) hypothesis is bounded by $|U|$.}
	\label{fig:experiment-int-interval-V}
\end{figure}
At last, we present another observation that explains why learning residual automata can be more efficient than deterministic ones, which applies to the non-symbolic case, too.
To answer each \tsc{mq} from a predicate learner, our algorithm uses $|V|$ MQs to the MAT.
Figure~\ref{fig:experiment-int-interval-V}(b) shows the average of $|V|$ in our experiments.
In the worst case analysis, $|V|$ may increase up to $O(n^2)$, but in most of these experiments, $|V|$ is much smaller than $n$, which keeps the number of MQs in our algorithm small.
An element $v$ is added to $V$ for denying at least one inclusion relation of a pair of rows, which occurs $O(n^2)$ times in the worst case.
In practice, 
one added element $v$ to $V$ may falsifies inclusions for many pairs of residual languages, while there is a lot of pairs between which inclusion properly hold, which will never been denied.
Therefore, $|V|$ tends to be much smaller than the worst case, and it saves many MQs by our algorithm.

\section*{Acknowledgement}
The research is supported by JSPS KAKENHI Grant Number JP18K11150.

\bibliographystyle{eptcs}
\bibliography{docs/ref}

\end{document}